\documentclass[conference,letterpaper]{IEEEtran}

\usepackage[utf8]{inputenc} 

\usepackage[T1]{fontenc}
\usepackage{url}
\usepackage{ifthen}
\usepackage{cite}
\usepackage[cmex10]{amsmath} 

\usepackage{amssymb}
\usepackage{xcolor}
\usepackage{graphicx}
\usepackage{url}
\usepackage{epsfig}
\usepackage{graphicx}
\usepackage{amsfonts}
\usepackage{amssymb}
\usepackage{amsbsy}
\usepackage{amsfonts}
\usepackage{amsthm, mathrsfs}
\usepackage{longtable}
\usepackage{stackengine}
\usepackage{mathtools}
\usepackage{stmaryrd}
\usepackage{relsize}
\usepackage{multirow}
\usepackage{bm}
\usepackage[T1]{fontenc}
\usepackage{url}
\usepackage{subcaption}
\usepackage{caption}
\usepackage{dsfont}
\usepackage{xfrac}
\usepackage{color}                    
\usepackage{xfrac}
 
\usepackage{setspace}

\usepackage[T1]{fontenc}         
\usepackage{amstext}
\usepackage{amssymb,amsmath,amsthm,enumitem}
\usepackage{sansmath}

\usepackage{mathrsfs}
\DeclareMathAlphabet{\mathpzc}{OT1}{pzc}{m}{it}

\def \H {\mathrm{H}} 
\def \I {\mathrm{I}} 
\def \D {\mathrm{D}} 

\theoremstyle{definition}

\newtheorem{proposition}{Proposition}
\newtheorem{definition}{Definition}
\newtheorem{lemma}{Lemma}
\newtheorem{remark}{Remark}

\usepackage{stfloats}

\def\markov{\hbox{$\--$}\kern-1.5pt\hbox{$\circ$}\kern-1.5pt\hbox{$\--$}}

\makeatletter
\newcommand*{\centernot}{%
  \mathpalette\@centernot
}
\def\@centernot#1#2{%
  \mathrel{%
    \rlap{%
      \settowidth\dimen@{$\m@th#1{#2}$}%
      \kern.5\dimen@
      \settowidth\dimen@{$\m@th#1=$}%
      \kern-.5\dimen@
      $\m@th#1\not$%
    }%
    {#2}%
  }%
}
\makeatother

\newcommand{\independent}{\perp\mkern-9.5mu\perp}
\newcommand{\notindependent}{\centernot{\independent}}

\begin{document}
\title{On Perfect Obfuscation:\\Local Information Geometry Analysis} 


 \author{%
   \IEEEauthorblockN{Behrooz~Razeghi\IEEEauthorrefmark{1},
   				     Flavio.~P.~Calmon\IEEEauthorrefmark{2}, 
                     Deniz~G\"{u}nd\"{u}z\IEEEauthorrefmark{3},
                     Slava~Voloshynovskiy\IEEEauthorrefmark{1}
                     }
   \IEEEauthorblockA{\IEEEauthorrefmark{1}%
                     University of Geneva
                   }
   \IEEEauthorblockA{\IEEEauthorrefmark{2}%
                     Harvard University}
  \IEEEauthorblockA{\IEEEauthorrefmark{3}%
                     Imperial College London}       
                     }

\maketitle

\begin{figure}[b]
\vspace{-0.4cm}
\parbox{\hsize}{\em
}\end{figure}

\begin{abstract}
We consider the problem of privacy-preserving data release for a specific utility task under perfect obfuscation constraint. We establish the necessary and sufficient condition to extract features of the original data that carry as much information about a utility attribute as possible, while not revealing any information about the sensitive attribute. This problem formulation generalizes both the information bottleneck and privacy funnel problems. We adopt a local information geometry analysis that provides useful insight into information coupling and trajectory construction of spherical perturbation of probability mass functions. This analysis allows us to construct the modal decomposition of the joint distributions, divergence transfer matrices, and mutual information. 
By decomposing the mutual information into orthogonal modes, we obtain the locally sufficient statistics for inferences about the utility attribute, while satisfying perfect obfuscation constraint. Furthermore, we develop the notion of perfect obfuscation based on $\chi^2$-divergence and Kullback–Leibler divergence in the Euclidean information space.


\end{abstract}



%
%
\section{Introduction}


Releasing an \textit{optimal representation} of data for a given task while simultaneously assuring \textit{privacy} of the individuals' identity and their associated data is one of the main challenges in the information-theory, signal processing, data mining and machine learning communities. 
An optimal representation is the most useful (sufficient), compressed (compact), and privacy-breaching (minimal) of data. 
Indeed, the optimal representation of data can be obtained subject to constraints on the target task and its computational and storage complexities.

We investigate the problem of privacy-preserving data release for a specific \textit{utility task} and consider an obfuscation-utility trade-off model where both utility and obfuscation are measured under logarithmic loss. 
Consider two communication parties, a data owner and a utility service provider. The data owner observes a random variable $X$ and acquires some utility, from the service provider, based on the information he discloses. 
Simultaneously, the data owner wishes to limit the amount of information revealed about a sensitive random variable $S$ that it depends on $X$. Therefore, instead of revealing $X$ directly to the service provider, the data owner releases a new representation, denoted by $Z$. 
The amount of information leaked to the service provider is measured by $\I \left( S; Z\right)$. 
In particular, the data owner is subjected to a constraint on the information complexity of representation that can be revealed to the service provider. This imposed information complexity is measured by $\I \left( X; Z \right)$. 
Moreover, in general, the utility acquired depends on a utility random variable $U$ that is dependent on $X$ and may be correlated to $S$. 
The amount of useful information revealed to the service provider is measured by $\I \left( U; Z \right)$. 
Therefore, considering Markov chain $\left( U, S \right) \markov X \markov Z$, 
our aim is to share a \textit{sanitized representation} $Z$ of \textit{observed data} $X$, through a stochastic mapping $P_{Z \mid X}$,  while preserving information about \textit{utility attribute} $U$ and obfuscate information about \textit{sensitive attribute} $S$. 
We called the stochastic mapping $P_{Z \mid X }$ the \textit{complexity-constraint obfuscation-utility-assuring mapping}. 


\vspace{5pt}

%
%

Information theoretic (IT) privacy approaches 
\cite{reed1973information,yamamoto1983source, evfimievski2003limiting, rebollo2009t, du2012privacy, sankar2013utility, calmon2013bounds, makhdoumi2013privacy, asoodeh2014notes, calmon2015fundamental, salamatian2015managing, basciftci2016privacy, asoodeh2016information, kalantari2017information, rassouli2018latent, asoodeh2018estimation, rassouli2018perfect, liao2018privacy, Hsu2019watchdogs, liao2019tunable,  sreekumar2019optimal, xiao2019maximal, diaz2019robustness, rassouli2019data, rassouli2019optimal}, 
model and analyze privacy-utility trade-offs using the IT metrics to provide 
asymptomatic 
or non-asymptotic 
privacy-utility-guaranteed frameworks. 
Inspired from \cite{yamamoto1983source}, in the most general form, the IT frameworks is based on the knowledge of specific `private' variable (or data, attribute, information) and correlated non-private variable, and assumption of exact joint distribution 
or partial statistical knowledge of private and/or non-private data.  
In this setup, the goal is to design a privacy assuring mapping that transforms the pair of these variables into a new representation that achieves a specific application-based target utility, while simultaneously minimizing the information inferred about the private variable. 
In many applications, the data $X$ is characterized over large (finite) alphabets while the attribute of interest, i.e., $U$, is characterized over small (finite) alphabets which results in $\I \left( U; X \right) \! \leq \! \H \left( U \right) \! \ll \! \H \left( X \right)$. 

\vspace{3pt}

%
%

%

Focusing on the finite alphabets and considering local information geometry analysis, we develop the notion of perfect obfuscation based on $\chi^2$-divergence and Kullback–Leibler (KL) divergence in the Euclidean information space.  
Under this analysis, we establish the necessary and sufficient condition to obtain representation $Z$ of the original data $X$ that maximizes the mutual information between utility attribute $U$ and released representation $Z$, while simultaneously revealing no information about sensitive attribute $S$. 
We decompose statistical dependence between random variables $U$, $S$, $X$ and $Z$ by decomposing the corresponding mutual informations $\I \left( X; Z\right)$, $\I \left( U; Z\right)$, and $\I \left( S; Z\right)$ into orthogonal modes. This model can be viewed as a generalization of two well-known bottleneck models, i.e., Information Bottleneck (IB) and Privacy Funnel (PF). 

 
\vfill

 \pagebreak

%
%
Throughout this paper, 
random variables are denoted by capital letters (e.g. $X$), deterministic values are denoted by small letters (e.g. $x$), alphabets (sets) are denoted by Calligraphic fonts (e.g. $\mathcal{X}$). 
Superscript $(.)^T$ stands for the transpose. 
For discrete random variable $X$, let consider a finite support set $\mathcal{X} \triangleq \{1, ..., \vert \mathcal{X} \vert \}$ with $2 \leq \vert \mathcal{X} \vert < + \infty$. 
We denote by $\mathcal{P}\left( \mathcal{X} \right)$ the set of all possible probability distributions of a random variable $X$ with range $\mathcal{X}$. 
We denote by $\mathbf{p}_X$ the probability mass function (pmf) vector with $i$-th entry equal to $p_X(i)$. 
%
%
$\H \left( \mathbf{p}_{X} \right) \coloneqq \mathbb{E}_{\mathbf{p}_{X}} \left[ - \log \mathbf{p}_{X} \right]$ denotes Shannon entropy. 
The relative entropy is defined as $\D_{\mathrm{KL}} \left( \mathbf{p}_{X} \Vert \mathbf{q}_{X} \right)  \coloneqq \mathbb{E}_{\mathbf{p}_{X}}  \big[ \log \frac{\mathbf{p}_{X}}{\mathbf{q}_{X}}  \big] $.


%
%
\section{Perfect Information Obfuscation Model}
\label{Sec:GeneralSetup}


Given the observed data $X$ the defender (data owner) wishes to release a representation $Z$ for a utility task $U$ while keeping another attribute $S$ as sensitive. 
Let us assume that $\mathbf{p}_{U, S, X}$ is fixed and known by both defender and adversary, and $\left( U, S\right) \markov X \markov Z$. We consider the non-interactive, one-shot regime, where the data owner discloses the representation $Z$ once, and no additional information is released. 
The general objective is to obtain stochastic map $\mathbf{P}_{\! Z \mid X}: \mathcal{X} \rightarrow \mathcal{Z}$ such that $\mathbf{p}_{U \mid Z=z} \approx \mathbf{p}_{U \mid X}, \forall z \! \!  \in \! \!  \mathcal{Z}, \forall \, U \!  \!  \in \! \!  \mathcal{U}, \forall \, X \! \!  \in \!  \!  \mathcal{X}$, while $\mathbf{p}_{S \mid Z=z} \approx
\mathbf{p}_{S}, \forall z \!  \in \!  \mathcal{Z}, \forall S \!  \in \!  \mathcal{S}$. 
This means that the posterior distribution of the utility attribute $U$ are similar given the released representation $Z$ and original data $X$, while the posterior distribution of the sensitive attribute $S$ are independent of the released representation $Z$. 
%
%
One can raises a question whether it is \textit{feasible} that the defender releases a representation $Z$ such that $\I \left( S; Z\right) =  \mathbb{E}_{\mathbf{p}_{Z}} \left[  \D_{\mathrm{KL}} \left( \mathbf{p}_{S\mid Z} \Vert \mathbf{p}_{S} \right) \right] = 0$, i.e., $S\independent Z$, while $\I \left( U; Z \right) > 0$, i.e., $U \notindependent Z$. 
This is a fundamental problem in information-theoretic privacy which is known as data disclosure under \textit{perfect privacy} regime. We will refer to this notion as \textit{perfect information obfuscation}. 

To pave our way, let us shortly review the previous models which are specific cases of our model. 
Consider the Markov chain $U \markov X \markov Z$. This gives us the celebrated Information Bottleneck (IB) problem \cite{tishby2000information}, where $\I \left( U; Z \right)$ is referred to as the useful released information (relevance of $U$) and $\I \left( X; Z \right)$ is referred to as the information complexity (description length). 
The goal of IB model is to find a representation $Z$ of $X$ such that $Z$ is maximally informative about $U$ while being minimally informative about $X$. 
We now consider the Markov chain $S \markov X \markov Z$. 
This gives us the well known Privacy Funnel (PF) problem \cite{makhdoumi2013privacy}, where $\I \left( S; Z \right)$ is referred to as the disclosed sensitive information, and $\I \left( X; Z \right)$ is referred to as the useful information. The goal of PF model is to obtain a representation $Z$ of $X$ that minimizes information between sensitive data $S$ and disclosed representation $Z$ while maximizes the amount of information between non-private (useful) data $X$ and disclosed representation $Z$.

Considering the PF model, the optimal obfuscation-utility coefficient for a given distribution $\mathbf{p}_{S,X}$ is defined as \cite{calmon2015fundamental}:\vspace{-2pt}
\begin{equation}\label{PerfectPrivacy_ContractionCoefficient}
\nu^{\ast} \left( \mathbf{p}_{S,X} \right) \coloneqq \mathop{\inf}_{\substack{\mathbf{p}_{Z \mid X}:\\
S \markov X\markov Z  }}
\frac{\I \left( S; Z \right)}{\I \left( X; Z \right)}. 
\end{equation}
They showed that $\nu^{\ast} \left( \mathbf{p}_{S,X} \right)$ is related to the smallest principal component of $\mathbf{p}_{S,X}$, and obtained the necessary and sufficient conditions under which $\nu^{\ast} \left( \mathbf{p}_{S,X} \right) = 0$. 
In \cite{rassouli2018perfect}, they studied a similar problem, however, they formulated the objective as of utility maximization under privacy leakage constraint. Hence, the optimal obfuscation-utility coefficient for a given distribution $\mathbf{p}_{S,X}$ is defined as:\vspace{-2pt}
\begin{eqnarray}\label{PerfectPrivacy_MaxU}
g_{\gamma} \! \left( \mathbf{p}_{S, X}\right) = \mathop{\sup}_{\substack{\mathbf{p}_{Z \mid X}:\\
\substack{ S \markov X\markov Z  \\ \I \left( S; Z \right) \leq \gamma }}} \I  \left( X; Z \right), 
\end{eqnarray}
where perfect information obfuscation is said to be feasible if $g_{0} \! \left( \mathbf{p}_{S, X}\! \right)   \! >\! 0$.

We consider the Markov model $\left( U, S\right) \markov X \markov Z$ which subsumes both IB and PF objectives. 
In this case, the functional \eqref{PerfectPrivacy_MaxU} can be generalized as:\vspace{-2pt}
\begin{eqnarray}\label{PIB_functional_PerfectPrivacy_ContractionCoefficient}
g_{\gamma} \! \left( \mathbf{p}_{U, S, X}\right) = \mathop{\sup}_{\substack{\mathbf{p}_{Z \mid X}:\\
\substack{\left( U, S\right) \markov X\markov Z  \\  \substack{\I \left( X; Z \right) \leq R} \\ \I \left( S; Z \right) \leq \gamma }}} \I  \left( U; Z \right). 
\end{eqnarray}
In particular, we study the necessary and sufficient conditions under which $g_{0} \! \left( \mathbf{p}_{U, S, X}\right) > 0$ under local information geometry analysis. 
To this goal, let us define the non-trivial perfect information obfuscation as follows.

\vspace{-2pt}

\begin{definition}[Non-trivial Perfect Information Obfuscation]
For a pair of random variables $\left( U, S, X\right)$, we say that non-trivial perfect information obfuscation is \textit{feasible} if there exists a random variable $Z$, that satisfies the following conditions:
\begin{itemize}
\item[1)]
$\left( U, S\right) \markov X \markov Z$ forms a Markov chain. 
\item[2)]
$S$ and $Z$ are independent, i.e., $S \independent Z$. 
\item[3)]
$U$ and $Z$ are not independent, i.e., $U \notindependent Z$. 
\end{itemize}
\end{definition}
Note that this definition subsumes the notion of perfect privacy addressed in \cite{rassouli2018perfect} as well as the notion of weakly independent introduced in \cite{berger1989multiterminal}.



%
We assume that we observe data $X$ and the distribution $\mathbf{p}_X$ is fixed. 
Hence, our purpose in non-trivial perfect information obfuscation problem is to construct a trajectory of perturbed pmfs such that a change along that direction changes $\mathbf{p}_{U}$, while keeps $\mathbf{p}_S$ unchanged. 
We establish the necessary and sufficient condition for the \textit{existence} of $\I \left( U; Z\right) > 0$, under a perfect obfuscation regime. 

\vspace{-3pt}

\begin{lemma}
Without loss of optimality we can restrict the size of $\mathcal{Z}$ in \eqref{PIB_functional_PerfectPrivacy_ContractionCoefficient} to $\vert \mathcal{Z} \vert \leq \vert \mathcal{X} \vert + 2$. 
\end{lemma}
\begin{proof}
The proof is based on Fenchel–Eggleston strengthening of Carath\'{e}odory's Theorem \cite{eggleston1958convexity}. 
\end{proof}
   

\section{Local Information Geometry Analysis}

To get insight into the trajectory construction, we adopt the local information geometry analysis \cite{huang2012linear, makur2015study, huang2019universal, makur2019information,  makurestimation2020} that provides geometrically appealing interpretation.  
Consider any reference pmf $\mathbf{p}_X \in \mathcal{P}^{\circ} \left( \mathcal{X}\right)$ in the relative interior of the probability simplex in $\mathbb{R}^{\vert \mathcal{X} \vert}$, where $\mathcal{P}^{\circ} \left( \mathcal{X}\right) \triangleq \{ \mathbf{p}_X \in \mathcal{P} \left( \mathcal{X} \right): p_X(x) > 0, \forall x \in \mathcal{X}  \}$ denotes the relative interior of $\mathcal{P} \left( \mathcal{X} \right)$. 
Consider a perturbed pmf $\mathbf{r}_X^{(\epsilon)} = \mathbf{p}_X + \epsilon \; \mathbf{h}_X \in \mathcal{P} \left( \mathcal{X}\right)$ from $\mathbf{p}_X$, for some small\footnote{We assume that $\epsilon \neq 0$ is small enough such that $\mathbf{r}_X^{(\epsilon)}$ is a valid pmf. Note that for larger values of $\epsilon$ it may not be entry-wise non-negative.} value $\epsilon$, where $\mathbf{h}_X$ is an additive perturbation vector of dimension $\vert \mathcal{X} \vert$, satisfying $\sum_x h_X (x) = 0$. The second order Taylor expansion of KL divergence can be written as:\vspace{-2pt}
\begin{subequations}
\begin{align}
\D_{\mathrm{KL}} \left( \mathbf{p}_{X}  \Vert \,  \mathbf{r}_{X}^{(\epsilon)}  \right) 
&=  - \sum_{x} p_X (x) \log \frac{r_X^{(\epsilon)}(x)}{p_X (x)} \\ 
&=   
 - \sum_{x} p_X (x) \log \left( 1 + \epsilon  \;  \frac{h_X(x)}{p_X (x)}  \right)
  \\ 
&=   
\! \frac{1}{2} \, \epsilon^2  \,  \sum_x \frac{1}{p_X(x)} h_X^2(x) + o \left( \epsilon^2 \right)  \label{Eq:KLweightedEuclidean} \\ 
&=   
\D_{\chi^2} \left( \mathbf{p}_{X}  \Vert \,  \mathbf{r}_{X}^{(\epsilon)}  \right) + o \left( \epsilon^2 \right). 
\end{align}
\end{subequations}
where $o \left( \epsilon^2 \right)$ denotes the Bachmann-Landau asymptotic little-$o$ notation\footnote{$\lim_{\epsilon \rightarrow 0} o \left( \epsilon^2 \right)/ \epsilon
^2 = 0$}, and $\D_{\mathcal{X}^2} \big( \mathbf{p}_{X}  \Vert \,  \mathbf{r}_{X}^{(\epsilon)}  \big) $ denotes ${\chi}^2$-divergence between $\mathbf{p}_{X} $ and $\mathbf{r}_{X}^{(\epsilon)} $, defined as follows:\vspace{-2pt}
\begin{equation}
\D_{\chi^2} \left( \mathbf{p}_{X}  \Vert \,  \mathbf{r}_{X}^{(\epsilon)}  \right) \triangleq \sum_{x\in \mathcal{X}} \frac{ \big( p_X(x) - r_X^{(\epsilon)} (x) \big)^2}{p_X (x)}. 
\end{equation}

Considering \eqref{Eq:KLweightedEuclidean}, one can view $\sum_x h_X^2(x) / p_X(x) $ as a weighted norm square of the perturbation vector $\mathbf{h}_X$, i.e., KL divergence is locally a weighted Euclidean metric\footnote{Note that all the well-defined $f$-divergences are locally equivalent to $\chi^2$-divergence measure to within a constant scale factor. Moreover, note that they locally behave like a Fisher information metric on the statistical manifold.}. 
Note that, in general, $\D_{\mathrm{KL}} \big( \mathbf{p}_{X}  \Vert \,  \mathbf{r}_{X}^{(\epsilon)}  \big) \neq \D_{\mathrm{KL}} \big( \mathbf{r}_{X}^{(\epsilon)}   \Vert \,   \mathbf{p}_{X}   \big)$, however, these divergences are equal up to the first order approximations, i.e., they are locally symmetric. 
Since by replacing the weights $p_X (x)$ in this norm by any other distribution in the neighborhood, the first order approximation remains the same.  
Therefore, we have $\D_{\mathrm{KL}} \big( \mathbf{p}_{X}  \Vert \,  \mathbf{r}_{X}^{(\epsilon)}  \big) = \D_{\mathrm{KL}} \big( \mathbf{r}_{X}^{(\epsilon)}   \Vert \,   \mathbf{p}_{X}   \big) + o \left( \epsilon^2 \right) $. This means that they resemble the standard Euclidean metric within a local neighborhood of pmfs around a reference pmf (i.e., from the center of the local neighborhood) in $\mathcal{P}^{\circ} \! \left( \mathcal{X} \right)$.

We now go one step further and instead of additive perturbation $\mathbf{r}_X^{(\epsilon)} = \mathbf{p}_X + \epsilon \; \mathbf{h}_X$, define the spherical perturbations for our analysis. 
Consider any reference pmf $\mathbf{p}_X \in \mathcal{P}^{\circ} \! \left( \mathcal{X}\right)$, and any other pmf $\mathbf{r}_X \in \mathcal{P} \left( \mathcal{X}\right)$. We can define the \textit{spherical perturbation} vector of $\mathbf{r}_X$ from $\mathbf{p}_X$ as $\mathbf{k}_X \triangleq \left( \mathbf{r}_X - \mathbf{p}_X \right) \mathsf{diag} \! \left( \sqrt{\mathbf{p}_X} \right)^{-1}$, 
where $\sqrt{\mathbf{p}_X}$ denotes the entry-wise square root of $\mathbf{p}_X$, and $\mathsf{diag} \left( \sqrt{\mathbf{p}_X} \right)$ denotes a diagonal matrix with principal entries equal to $\sqrt{\mathbf{p}_X}$. Now, we can construct a trajectory of spherically perturbation pmfs as follows:\vspace{-2pt}
\begin{subequations}
\begin{align}
\mathbf{r}_X^{(\epsilon)} &= \mathbf{p}_X + \epsilon \; \mathbf{k}_X \, \mathsf{diag}\! \left( \sqrt{\mathbf{p}_X} \right) \label{Eq:SphericallyPerturbation} \\
&= \left( 1 - \epsilon \right) \mathbf{p}_X + \epsilon \, \mathbf{r}_X,
\end{align}
\end{subequations}
where $\epsilon \in \left( 0, 1 \right)$ controls closeness of $\mathbf{r}_X^{(\epsilon)}$ and $\mathbf{p}_X$. The second equation expresses $\mathbf{r}_X^{(\epsilon)}$ as a convex combination of $\mathbf{p}_X$ and $\mathbf{r}_X$. 
Note that $\mathbf{k}_X$ in \eqref{Eq:SphericallyPerturbation} is a normalized perturbation vector and provides the direction of our trajectory. Furthermore, considering the constraint $\sum_x h_X (x) = 0$ we can verify that $\mathbf{k}_X$ in \eqref{Eq:SphericallyPerturbation} satisfies the orthogonality constraint $\mathrm{(C1)}:$ $\mathbf{k}_X^T \sqrt{\mathbf{p}_X}  = 0$. Finally, we can rewrite the quadratic approximation of KL divergence as a scaled Euclidean norm of $\mathbf{k}_X$. We have:\vspace{-2pt}
\begin{equation}
\D_{\mathrm{KL}} \left( \mathbf{p}_{X}  \Vert \,  \mathbf{r}_{X}^{(\epsilon)}  \right) = 
\frac{1}{2} \, \epsilon^2 \,  {\Vert \mathbf{k}_X \Vert}_2^2
+ o \left( \epsilon^2 \right). 
\end{equation}
Note that using this local approximation we can construct inner products as well as orthogonal perturbations and projections in the Euclidean space. 

\begin{remark}
In \cite{berger1989multiterminal}, the authors defined the notion of \textit{weakly independence} for a pair of random variables $\left( S, X \right) \in \mathcal{S} \times \mathcal{X}$ ($\vert \mathcal{S} \vert , \vert \mathcal{X} \vert < \infty$) as existence of a random variable $Z$ such that: (i) $S \markov X \markov Z$ forms a Markov chain, (ii) $S$ and $Z$ are independent, and (iii) $X$ and $Z$ are not independent. They showed that such a random variable $Z$ exists if and only if the columns of $\mathbf{P}_{S \mid X}$ are linearly dependent. Inspired by this notion of weakly independent, the authors in \cite{rassouli2018perfect, rassouli2019data} carefully studied and analyzed perfect obfuscation problem where the goal is to release the useful information $X$ while keeping $S$ as private. Here we extend both, and establish the notion of weakly dependence based on $\mathrm{KL}$-divergence and $\chi^2$-divergence. 
\end{remark}

Using the local information approximation, we can write the conditional distributions $\mathbf{p}_{X \mid Z =z}$ as perturbation of $\mathbf{p}_X$, i.e., we have:\vspace{-2pt}
\begin{equation}\label{Eq:SphericallyPerturbation_Conditional}
\mathbf{p}_{X \mid Z =z} = \mathbf{p}_X + \epsilon \; \mathbf{k}_{X \mid z} \, \mathsf{diag}\! \left( \sqrt{\mathbf{p}_X} \right) .
\end{equation}
We just need to ensure that $\mathbf{p}_{X \mid Z =z}$, for different values $z$, be a valid probability distribution and satisfy the marginal constraints. 
Hence, we additionally required $(\mathrm{C2}):$ $\sum_z p_Z(z)\,  k_{X \mid z} (x) \sqrt{p_X (x)}= 0, \forall x \in \mathcal{X}$ which guarantees that marginal pmf of $X$ is preserved, i.e., $\sum_z p_Z (z) p_{X \mid Z=z} = p_X$.
Therefore, our purpose in non-trivial perfect obfuscation problem under local information geometry analysis is to design the latent distribution $\mathbf{p}_Z$ and the conditional distributions $\mathbf{p}_{X \mid Z=z}$, for different values of $z$, such that: (i) the constraints $(\mathrm{C1})$ and $(\mathrm{C2})$ are satisfied, (ii) $S \independent Z$, and (iii) $U \notindependent  Z$.

\begin{proposition}\label{proposition_PerfectObfuscation1}
For perfect obfuscation data released model $\left( U, S \right) \markov X \markov Z$ under local information geometry analysis, the non-trivial perfect obfuscation is feasible if and only if for all $z \in \mathcal{Z}$ we simultaneously have:
\begin{subequations}\label{Eq:proposition_PerfectObfuscation1}
\begin{align}
\mathbf{W}_S \;  \mathbf{k}_{X\mid z} \, \mathsf{diag}\! \left( \sqrt{\mathbf{p}_X} \right) = \boldsymbol{0}, \\
\mathbf{W}_U \;  \mathbf{k}_{X\mid z} \, \mathsf{diag}\! \left( \sqrt{\mathbf{p}_X} \right) \neq \boldsymbol{0}, 
\end{align}
\end{subequations}
where $\mathbf{W}_S \coloneqq \ \mathbf{P}_{S \mid X} : \mathcal{X} \rightarrow \mathcal{S}$ and $\mathbf{W}_U \coloneqq \ \mathbf{P}_{U \mid X} : \mathcal{X} \rightarrow \mathcal{U}$ are fixed probability transition kernels, with dimension $\vert \mathcal{S} \vert \times \vert \mathcal{X} \vert$ and $\vert \mathcal{U} \vert \times \vert \mathcal{X} \vert$, respectively. 
\end{proposition}
\begin{proof}
To ensure perfect obfuscation, we need $\mathbf{p}_{S \mid Z=z} = \mathbf{p}_S, \forall S \in \mathcal{S}, z \in \mathcal{Z}$. Considering the Markov chain $S \markov  X \markov Z$, we have:
\begin{eqnarray}
\mathbf{p}_{S \mid Z =z} \!\! \! &=& \!\!\! \mathbf{W}_S \, \mathbf{p}_{X \mid Z=z} = \mathbf{W}_{S} \;  \mathbf{p}_X + \epsilon \, \mathbf{W}_S \; \mathbf{h}_{X\mid z}  \nonumber \\
\!\! \! &=& \!\!\!
\mathbf{p}_S  + \epsilon \; \mathbf{W}_S \; \mathbf{k}_{X \mid z} \, \mathsf{diag}\! \left( \sqrt{\mathbf{p}_X} \right), \; \forall z \in \mathcal{Z}. 
\end{eqnarray}
Therefore, $ S \independent Z$, if and only if $\mathbf{W}_S \; \mathbf{k}_{X \mid z} \, \mathsf{diag}\! \left( \sqrt{\mathbf{p}_X} \right) = \boldsymbol{0}, \forall z \in \mathcal{Z}$. %
Analogously, considering the Markov chain $U \markov  X \markov Z$, we have:
\begin{eqnarray}
\mathbf{p}_{U \mid Z =z} \!\! \! &=& \!\!\! \mathbf{W}_U \, \mathbf{p}_{X \mid Z=z} = \mathbf{W}_{U} \;  \mathbf{p}_X + \epsilon \, \mathbf{W}_U \; \mathbf{h}_{X \mid z}  \nonumber \\
 \!\!\! &=& \!\!\!
\mathbf{p}_U  + \epsilon \; \mathbf{W}_U \; \mathbf{k}_{X\mid z} \, \mathsf{diag}\! \left( \sqrt{\mathbf{p}_X} \right),  \; \forall z \in \mathcal{Z}. 
\end{eqnarray}
Hence, if we can find the perturbation direction such that $\mathbf{W}_S \;  \mathbf{k}_{X\mid z} \, \mathsf{diag}\! \left( \sqrt{\mathbf{p}_X} \right) = \boldsymbol{0}$  and  $\mathbf{W}_U \; \mathbf{k}_{X \mid z} \, \mathsf{diag}\! \left( \sqrt{\mathbf{p}_X} \right) \neq \boldsymbol{0}$, for some $z \in \mathcal{Z}$, the non-trivial solution, i.e., $\I \left( U; Z \right) > 0$, is possible.  
Conversely, we have a non-trivial solution only if there exists a random variable $Z$ and a valid perturbation vector $\mathbf{k}_{X \mid z}$ such that a change along that direction changes $\mathbf{p}_U$, while keeping $\mathbf{p}_S$ unchanged. This implies \eqref{Eq:proposition_PerfectObfuscation1}.
\end{proof}

%

\begin{definition}[Divergence Transfer Matrix]\label{Def:DivergenceTransferMatrix}
Given the random variables $X\in \mathcal{X}$ and $Z \in \mathcal{Z}$ with joint pmf $\mathbf{P}_{X, Z} \in \mathcal{P} \left( \mathcal{X} \times \mathcal{Z} \right)$, with conditional pmfs $\mathbf{P}_{X \mid Z} \in \mathcal{P} \left( \mathcal{X} \mid \mathcal{Z} \right)$ and marginal pmfs satisfying $\mathbf{p}_X \in \mathcal{P}^{\circ} \left( \mathcal{X} \right)$ and $\mathbf{p}_Z \in \mathcal{P}^{\circ} \left( \mathcal{X} \right)$, the divergence transfer matrix associated with 
$\mathbf{P}_{X,Z}$ is defined as follows:\vspace{-2pt}
\begin{eqnarray}
\mathbf{B}_{X,Z} \! = \mathbf{B} \left( \mathbf{P}_{X,Z} \right)\!\! \! \!\! & \triangleq & \! \!\!\!\!
 \mathsf{diag}\! \left( \sqrt{\mathbf{p}_X} \right)^{-1} \, \mathbf{P}_{X, Z} \, \mathsf{diag}\! \left( \sqrt{\mathbf{p}_Z} \right)^{-1} \nonumber \\
\!\! & = & \! \!\!\!\! \mathsf{diag}\! \left( \sqrt{\mathbf{p}_X} \right)^{-1} \, \mathbf{P}_{X \mid Z} \, \mathsf{diag}\! \left( \sqrt{\mathbf{p}_Z} \right)\! . \label{DTM_XZ}
\end{eqnarray}
\end{definition}

Note that based on the above definition $\mathbf{B}_{X,Z}^T = \mathsf{diag}\! \left( \sqrt{\mathbf{p}_Z} \right)^{-1} \, \mathbf{P}_{Z \mid X} \, \mathsf{diag}\! \left( \sqrt{\mathbf{p}_X} \right)$. 
We now express the Singular Value decomposition (SVD) of $\mathbf{B}_{X,Z}$ as:
\begin{equation}\label{DTM_XZ_SVD}
\mathbf{B}_{X,Z} = \sum_{i=1}^{K} \sigma_i^{XZ} \,  \boldsymbol{\psi}_i^Z   \, {(\boldsymbol{\psi}_i^Z )}^T, 
\end{equation}
where $K \coloneqq \min \left\{ \vert \mathcal{X} \vert , \vert \mathcal{Z} \vert \right\}$, $\sigma_i^{XZ}$ denotes the $i$-th singular value, and where $\boldsymbol{\psi}_i^Z$ and $\boldsymbol{\psi}_i^X$ are the corresponding left (output) and right (input) singular vectors. By convention, suppose that $\sigma_1^{XZ} \geq \sigma_2^{XZ} \geq \cdots \geq \sigma_K^{XZ}$. Likewise consider SVD of $\mathbf{B}_{U,X}$, $\mathbf{B}_{S,X}$, $\mathbf{B}_{U,Z}$ and $\mathbf{B}_{S,Z}$. 

\begin{proposition}[Local Approximation of Information Measures]\label{Proposition_LocalApprox}
Under the local approximation conditions, the information complexity $\I \left( X; Z \right)$, utility information $\I \left( U; Z \right)$, and information leakage $\I \left( S; Z \right)$ can recast as:
\begin{subequations}\label{Eq:LocalInfoApprox}
\begin{align}
\I \left( X; Z \right) & =   \frac{1}{2} \, \epsilon^2 \, \sum_{z \in \mathcal{Z}} p_Z(z) \,  {\Vert  \, \mathbf{k}_{X\mid z}  \, \Vert}_2^2  + o \left( \epsilon^2 \right) \\
& =
 \frac{1}{2} \left(  {\Vert  \mathbf{B}_{X,Z}  \Vert}_{\mathrm{F}}^2 - 1 \right) + o \left( \epsilon^2 \right) ,   \\
\I \left( U; Z \right) & =    \frac{1}{2} \, \epsilon^2 \, \sum_{z \in \mathcal{Z}} p_Z(z) \,  {\Vert  \mathbf{B}_{U,X} \,  \mathbf{k}_{X \mid z} \Vert}_2^2   + o \left( \epsilon^2 \right) \\
& = \frac{
1}{2} \left(  {\Vert  \mathbf{B}_{U,Z}  \Vert}_{\mathrm{F}}^2 - 1 \right) + o \left( \epsilon^2 \right), \\
\I \left( S; Z \right)  &=    \frac{1}{2} \, \epsilon^2 \, \sum_{z \in \mathcal{Z}} p_Z(z) \,  {\Vert  \mathbf{B}_{S,X} \,  \mathbf{k}_{X \mid z} \Vert}_2^2  + o \left( \epsilon^2 \right) \\ 
& = \frac{1}{2} \left(  {\Vert  \mathbf{B}_{S,Z}  \Vert}_{\mathrm{F}}^2 - 1 \right) + o \left( \epsilon^2 \right) ,
\end{align}
\end{subequations}
where $\mathbf{B}_{U,X}$ and $\mathbf{B}_{U,Z}$ are defined analogous to \eqref{DTM_XZ}. 
\end{proposition}
\begin{proof}
See Appendix~\ref{Appendix_ProofLocalApprox}.
\end{proof}

The local approximation \eqref{Eq:LocalInfoApprox} gives a nice geometric interpretation. Consider a local divergence sphere in $\mathcal{P}\left( \mathcal{X} \right)$ constructed as \eqref{Eq:SphericallyPerturbation_Conditional}. 
The divergence transfer matrix $\mathbf{B}_{X,Z}$ maps a local divergence sphere in $\mathcal{P}\left( \mathcal{X} \right)$ to a local divergence ellipsoid in $\mathcal{P} \left( \mathcal{Z} \right)$ (Fig.~\ref{Fig:DTM_map_SVD}). 
%
%
Noting that the Markov chain $U \markov X \markov Z$ implies:
\begin{subequations}
\begin{align}
p_{U} (u) & = \sum_{x \in \mathcal{X}} p_{U \mid X} (u \! \mid \! x) \, p_X(x), \\
p_{U \mid Z} (u \! \mid \! z) &= \sum_{x \in \mathcal{X}} p_{U \mid X} (u \! \mid \! x) \, p_{X \mid Z} (x \! \mid \! z). 
\end{align}
\end{subequations}
Analogously, consider the likewise relation for $p_{S} (s) $ and $p_{S \mid Z} (s \!\!\!  \mid \! \! \! z)$. 
Therefore, the information perturbation vector $\mathbf{k}_{X \mid z}$ maps to the associated information perturbation vectors $\mathbf{k}_{U \mid z} \triangleq \mathbf{B}_{U,X} \mathbf{k}_{X \mid z}$ and $\mathbf{k}_{S \mid z} \triangleq \mathbf{B}_{S,X} \mathbf{k}_{X \mid z}$. 
In other words, the local geometry of $p_{X \mid Z}$ in the simplex $\mathcal{P} \left( \mathcal{X} \right)$ induces a corresponding local geometry for $p_{U \mid Z}$ and $p_{S \mid Z}$. 

Under the local information approximation \eqref{Eq:LocalInfoApprox} and satisfying the constraints $(\mathrm{C}1)$ and $(\mathrm{C}2)$ of perturbation construction, and neglecting $o \left( \epsilon^2 \right)$ terms, the optimization problem under perfect obfuscation constraint can recast as:
\begin{subequations}\label{PerfectObfuscationOptLocalProblem}
\begin{align}
\mathop{\max}_{\mathbf{p}_Z, \mathbf{p}_{X \mid Z}} \qquad &
\sum_{z\in \mathcal{Z}} p_Z(z) \, {\Vert  \mathbf{B}_{U,X} \,  \mathbf{k}_{X \mid z} \Vert}_2^2 \label{PerfectObfuscationOptLocalProblem1} \\
 \mathrm{s.t.}     \qquad  &
\sum_{z\in \mathcal{Z}} p_Z(z) \, {\Vert \mathbf{k}_{X\mid z}  \Vert}_2^2 \leq R^\prime , \label{PerfectObfuscationOptLocalProblem2} \\
\quad \quad  \;\; & \sum_{z\in \mathcal{Z}} p_Z(z) \, {\Vert  \mathbf{B}_{S,X} \,  \mathbf{k}_{X \mid z} \Vert}_2^2  = 0, \label{PerfectObfuscationOptLocalProblem3} 
\end{align}
\end{subequations}
or equivalently, as:
\begin{equation}\label{PerfectObfuscationOptLocalProblemFrobenius}
\mathop{\max}_{\mathbf{p}_Z, \mathbf{p}_{X \mid Z}}  
{\Vert \mathbf{B}_{U,Z}  \Vert}_{\mathrm{F}}^2  
  \;\; \;  \mathrm{s.t.} \;\;    
{\Vert \mathbf{B}_{X,Z}  \Vert}_{\mathrm{F}}^2 \leq R^{\prime \prime}, \; 
{\Vert \mathbf{B}_{S,Z}  \Vert}_{\mathrm{F}}^2  =1, 
\end{equation}
where $R^\prime = 2R/ \epsilon^2$ and $R^{\prime \prime}= 2R+1$. 
Note that the pmf $\mathbf{p}_Z$ does not affect the optimization and can be removed from \eqref{PerfectObfuscationOptLocalProblem}. 
Hence SVD solves the optimization problem \eqref{PerfectObfuscationOptLocalProblem} by finding $\mathbf{p}_{X \mid Z}$. 
Finally, note that by construction $\I \left( X; Z \right) \leq \frac{1}{2} \epsilon^2$.  Hence, as long as $R \leq \frac{1}{2} \epsilon^2$, we can relax the associated constraint in local information geometry analysis.

\begin{figure}
\centering
\includegraphics[scale=0.67]{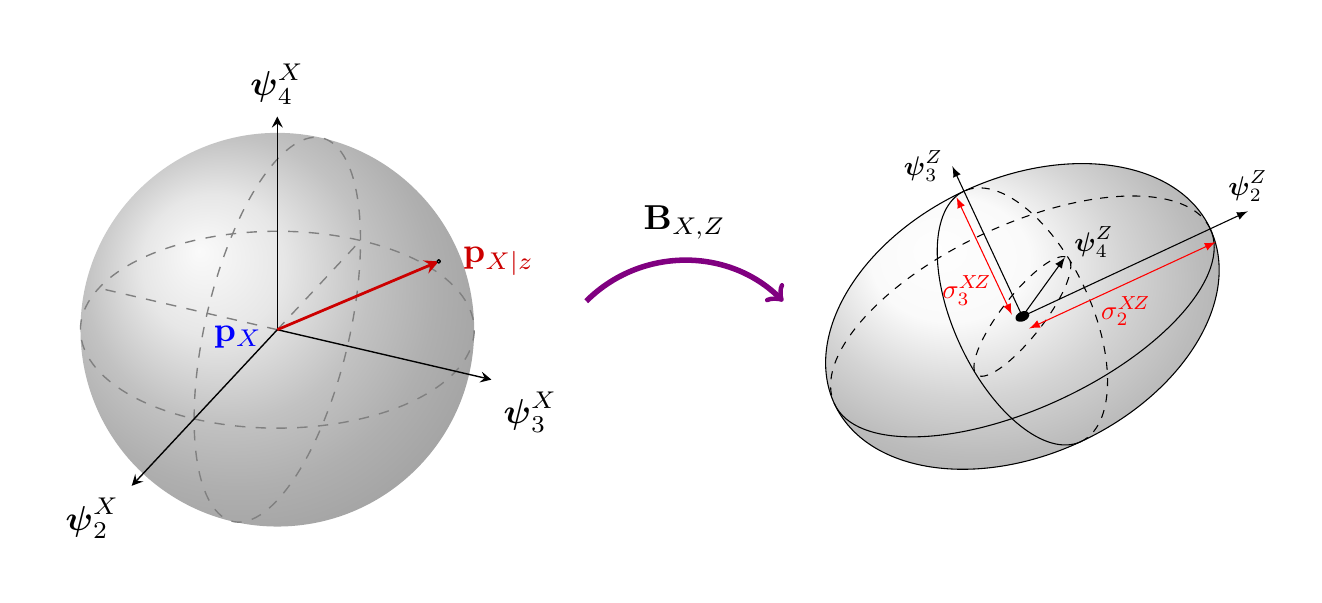}
\caption{The information geometry associated with the divergence transfer matrix $\mathbf{B}_{X,Z}$. Visualization for second, third and forth singular vectors, ignoring the first invalid direction. $\mathbf{B}_{X,Z}$ maps a local divergence sphere in $\mathcal{P}\! \left( \mathcal{X} \right)$ to a local divergence ellipsoid in $\mathcal{P} \! \left( \mathcal{Z} \right)$. }
\label{Fig:DTM_map_SVD}
\end{figure}

\pagebreak

To get insight of the optimization problem \eqref{PerfectObfuscationOptLocalProblem}, let us ignore the constraints \eqref{PerfectObfuscationOptLocalProblem2} and \eqref{PerfectObfuscationOptLocalProblem3}. Note that based on the constraint ($\mathrm{C1}$) the valid normalized perturbation $\mathbf{k}_{X \mid z}$ must be orthogonal to $\sqrt{\mathbf{p}_X}$, hence $\sqrt{\mathbf{p}_X}$ (the right singular vector of $\mathbf{B}_{U,X}$ corresponding to the largest singular value) is an invalid direction to perturb pmf. Letting $\sigma_2^{\mathrm{UX}}$ be the second largest singular value of $\mathbf{B}_{U,X}$, we have ${\Vert  \mathbf{B}_{U,X} \,  \mathbf{k}_{X \mid z} \Vert}^2 \leq {(\sigma_2^{\mathrm{UX}})}^2 \; {\Vert  \mathbf{k}_{X \mid z} \Vert}^2$. Therefore, under this assumption, the optimal solution to \eqref{PerfectObfuscationOptLocalProblem1} is to choose the perturbation $\mathbf{k}_{X \mid z}$ to be along the right singular vector of $\mathbf{B}_{U,X}$ corresponding to the second largest singular value. Note that all the unit norm right singular vectors of $\mathbf{B}_{U,X}$ which are orthogonal to $\sqrt{\mathbf{p}_X}$ are a valid perturbation. This means that any linear combination of these singular vectors also valid candidates for $\{ \mathbf{k}_{X \mid z}  , z \in \mathcal{Z}\}$.


\vspace{5pt}

Let $\mathsf{Range} \left( \mathbf{B}_{U,X}  \right)= \{ \mathbf{B}_{U,X} \, \mathbf{k}_{X\mid z} \mid \mathbf{k}_{X\mid z}  \in \mathbb{R}^{\vert \mathcal{X} \vert} \}  \subseteq \mathbb{R}^{\vert \mathcal{U} \vert}$ denotes the range-space of $\mathbf{B}_{U,X}$,
and $\mathsf{Null} \left( \mathbf{B}_{S,X} \right) = \{ \mathbf{k}_{X\mid z} \mid \mathbf{B}_{S,X} \mathbf{k}_{X\mid z} = \boldsymbol{0} \} \subseteq \mathbb{R}^{\vert \mathcal{X} \vert}$ denotes the null-space of $\mathbf{B}_{S,X}$. We now have the following proposition.

\begin{proposition}
For perfect obfuscation data released model $\left( U\! , S \right) \markov X \markov Z$ under local information geometry analysis, the non-trivial perfect information obfuscation is feasible if and only if:
\begin{equation}
\mathrm{dim} \left( \, \mathsf{Range}\left( \mathbf{B}_{U,X} \right) \cap \mathsf{Null}\left( \mathbf{B}_{S,X}  \right) \,   \right) > 0. 
\end{equation}
\end{proposition}
\begin{proof}
The proof follows by using Proposition~\ref{proposition_PerfectObfuscation1} and noting that Divergence Transfer Matrix $\mathbf{B}_{U,X}$ (likewise $\mathbf{B}_{S,X}$) is an equivalent representation for $\mathbf{P}_{U,X}$ (likewise $\mathbf{P}_{S,X}$) and, in turn, $\mathbf{P}_{U \mid X}$ (likewise $\mathbf{P}_{S \mid X}$). 
\end{proof}



We now relate the solutions of \eqref{PerfectObfuscationOptLocalProblem} to locally sufficient statistics for inferences about utility attribute $U$ based on $Z$. 
Let us consider an arbitrary embedding (feature) $f: \mathcal{X} \rightarrow \mathbb{R}$ and let $g: \mathcal{Z} \rightarrow \mathbb{R}$ be the embedding (feature) induced by $f$ through conditional expectation with respect to $\mathbf{p}_{X\mid Z=z}$. We have:
\begin{equation}\label{Eq:conditional_expectation}
g(z) = \mathbb{E} \left[  f(X)  \mid Z = z \right], \quad z \in \mathcal{Z}. 
\end{equation}
We can recast \eqref{Eq:conditional_expectation} as:
\begin{eqnarray}\label{Eq:conditional_expectation_byDTM}
g (z) \!\! &=& \!\!  \frac{1}{p_Z(z)} \sum_{x \in \mathcal{X}} p_{X,Z} (x, z) f(x) \nonumber \\
&=& \!\!  \frac{1}{\sqrt{p_Z (z)}} \sum_{x \in \mathcal{X}} B_{X,Z} (x,z) \sqrt{p_X(x)} f(x),  
\end{eqnarray}
where $B_{\!X,Z} (x,z) \!\!=\! \!\frac{p_{X,Z} (x,z)}{\sqrt{p_{X}(x)} \sqrt{p_Z(z)}}, \forall x \!  \in \! \mathcal{X}, z  \! \in  \! \mathcal{Z}$ is the $(x,z)$-th entry of $\mathbf{B}_{X,Z}$. 
We now define $\xi^{X} (x) \coloneqq \sqrt{p_X(x)} f(x)$ and $\xi^{Z} (z) \coloneqq \sqrt{p_Z(z)} g(z)$, $\forall x \in \mathcal{X}, z \in \mathcal{Z}$. 
Then we can express \eqref{Eq:conditional_expectation_byDTM} as:
\begin{equation}
\xi^Z (z) = \sum_{x \in \mathcal{X}} B_{X,Z} (x,z)\,  \xi^X (x). 
\end{equation}
The vectors $\boldsymbol{\xi}^X$ and $\boldsymbol{\xi}^Z$ whose $x$-th and $z$-th entries are $\xi^{X} (x), \forall x \in \mathcal{X}$ and $\xi^{Z} (z), \forall z \in \mathcal{Z}$, respectively, can be referred as feature vectors associated with the feature functions $f$ and $g$. 

According to \eqref{DTM_XZ_SVD} and proof of proposition~\ref{Proposition_LocalApprox}, we have $B_{X,Z} (x,z) = \sum_{i=1}^{K} \sigma_i^{XZ} \psi_i^X (x) \psi_i^Z (z) = \sqrt{p_X(x)} \sqrt{p_Z(z)} + \sum_{i=2}^{K} \sigma_i^{XZ} \psi_i^X (x) \psi_i^Z (z)$. We now define features $f_i^\ast : \mathcal{X} \rightarrow \mathbb{R}$ and $g_i^\ast : \mathcal{X} \rightarrow \mathbb{R}$, for $i = 2, 3, ..., K$, as follows:
\begin{subequations}
\begin{align}
f_i^\ast (x)   \coloneqq \frac{\psi_i^X(x) }{\sqrt{p_X(x)}}, \\
g_i^\ast (z)  \coloneqq \frac{\psi_i^Z(z) }{\sqrt{p_Z(z)}}. 
\end{align}
\end{subequations}
Hence we have:
\begin{equation}
B_{X,Z}(x,z) = \sqrt{p_X(x)} \sqrt{p_Z(z)} \bigg( 1 + \sum_{i=2}^{K} \sigma_i^{XZ} f_i^\ast (x)  g_i^\ast (z)   \bigg). 
\end{equation}
Noting that $P_{X,Z} (x,z) = B_{X,Z} (x,z) \sqrt{p_X(x)} \sqrt{p_Z(z)} = p_X(x) p_Z(z) \big( 1 + \sum_{i=2}^{K} \sigma_i^{XZ} f_i^\ast (x)  g_i^\ast (z) \big)$, we have modal decomposition of joint distributions, conditional distributions, and mutual information in terms of feature functions $\left( f_i^\ast, g_i^\ast \right), i = 2, 3, ..., K$. 
Hence, the valid perturbation directions in optimization problem \eqref{PerfectObfuscationOptLocalProblem} give us the corresponding valid feature functions, as well as, associated locally normalized sufficient statistics for inferences about $U$ based on $Z$, under perfect obfuscation constraint.




\section{Conclusion}

Adopting a local information geometry analysis and considering mutual information as both obfuscation and utility measure, we studied a data released mechanism for a given utility task, and under perfect obfuscation constraint. The addressed model subsumes both the Information Bottleneck model and the Privacy Funnel model. 
We studied the notion of perfect obfuscation based on $\chi^2$-divergence and Kullback–Leibler divergence in the Euclidean information space.
Furthermore, we characterized the necessary and sufficient conditions under which a non-trivial solution is feasible. 





\appendices

\section{Proof of Proposition~\ref{Proposition_LocalApprox}}
\label{Appendix_ProofLocalApprox}
\begin{proof}
\begin{subequations}
\begin{align} 
 \I \left( X; Z \right)   &=    \sum_{z} p_Z(z) \,  \D_{\mathrm{KL}} \left( \mathbf{p}_{X \mid Z=z}\,  \Vert \,  \mathbf{p}_X \right) \\
 &= \frac{1}{2} \, \epsilon^2 \sum_{z\in \mathcal{Z}} p_Z(z) \, {\Vert \mathbf{k}_{X\mid z}  \Vert}_2^2 + o \left( \epsilon^2 \right)
 \\
&=     \frac{1}{2} \epsilon^2 \sum_{z,x} p_Z(z) {\left( \frac{p_{X \mid Z} (x \mid z) - p_X(x)}{\epsilon \, \sqrt{p_X(x)}} \right)}^2 + o \left( \epsilon^2 \right) 
 \\
&=   \frac{1}{2} \sum_{z, x} {\left(  \frac{p_{X, Z} (x, z) - p_X(x) p_Z(z)}{\sqrt{p_X(x)} \sqrt{p_Z(z)}} \right) }^2 + o \left( \epsilon^2 \right)
  \\
&=  
\frac{1}{2} { \Vert \mathbf{B}_{X,Z} - \sqrt{\mathbf{p}_X} \sqrt{\mathbf{p}_Z}^T   \Vert  }_{\mathrm{F}}^2 + o \left( \epsilon^2 \right)
  \\
&=   \frac{1}{2} \left(  {\Vert \mathbf{B}_{X,Z}  \Vert}_{\mathrm{F}}^2  - 1 \right)+ o \left( \epsilon^2 \right),\label{I_XZ_Frob_DTM}
\end{align}
\end{subequations}

\pagebreak

\begin{subequations}\label{Eq:I_UZ_local_proof}
\begin{align}
\I \left( U; Z \right)  & =   \sum_{z} p_Z(z) \,  \D_{\mathrm{KL}} \left( \mathbf{p}_{U \mid Z=z}\,  \Vert \,  \mathbf{p}_U \right) \! \\
& =  
 \frac{1}{2} \, \epsilon^2 \sum_{z\in \mathcal{Z}} p_Z(z) \cdot \nonumber \\
 &\qquad  {\Vert \mathsf{diag}\! \left( \! \sqrt{\mathbf{p}_U} \right)^{-1} \! \mathbf{W}_U \mathsf{diag}\! \left(\! \sqrt{\mathbf{p}_X} \right) \mathbf{k}_{X\mid z}  \Vert}_2^2\!  +   o\! \left( \epsilon^2 \right)   \\
& =  
 \frac{1}{2} \, \epsilon^2 \sum_{z\in \mathcal{Z}} p_Z(z) \, {\Vert  \mathbf{B}_{U,X} \,  \mathbf{k}_{X \mid z} \Vert}_2^2 + o \left( \epsilon^2 \right) 
 \\
&=   \frac{1}{2} \epsilon^2 \sum_{z,u} p_Z(z) {\left( \frac{p_{U \mid Z} (u \mid z) - p_U(u)}{\epsilon \, \sqrt{p_U (u)}} \right)}^2 + o \left( \epsilon^2 \right)
  \\
&=    \frac{1}{2} \sum_{z, x} {\left(  \frac{p_{X, Z} (x, z) - p_X(x) p_Z(z)}{\sqrt{p_X(x)} \sqrt{p_Z(z)}} \right) }^2 + o \left( \epsilon^2 \right)
  \\
&=  
\frac{1}{2} { \Vert \mathbf{B}_{U,Z} - \sqrt{\mathbf{p}_U} \sqrt{\mathbf{p}_Z}^T   \Vert  }_{\mathrm{F}}^2 
+ o \left( \epsilon^2 \right)
  \\
&=  
 \frac{1}{2} \left(  {\Vert \mathbf{B}_{U,Z}  \Vert}_{\mathrm{F}}^2  - 1 \right)+ o \left( \epsilon^2 \right). \label{I_UZ_Frob_DTM}
\end{align}
\end{subequations}
%
The equalities \eqref{I_XZ_Frob_DTM} and \eqref{I_UZ_Frob_DTM} follow by noticing that the largest singular value of divergence transfer matrices $\mathbf{B}_{X,Z}$ and $\mathbf{B}_{U,Z}$ are 1, i.e., their spectral norm is equal to one. 
Note that $\mathbf{B}_{X,Z}$ and $\mathbf{B}_{U,Z}$ originate from the column stochastic transition matrices of conditional probabilities $\mathbf{P}_{X \mid Z}$ and $\mathbf{P}_{U \mid Z}$, respectively. 
Therefore, the corresponding right (input) singular vectors
are as follows:
\begin{subequations}
\begin{align}
\boldsymbol{\psi}_1^X = \mathbf{B}_{X,Z} \, \sqrt{\mathbf{p}_Z} & = \sigma_1^{\mathrm{XZ}}  \sqrt{\mathbf{p}_X} = \sqrt{\mathbf{p}_X},
\\
\boldsymbol{\psi}_1^U = \mathbf{B}_{U,Z} \, \sqrt{\mathbf{p}_Z} & = \sigma_1^{\mathrm{UZ}} \sqrt{\mathbf{p}_U}  = \sqrt{\mathbf{p}_U}.
\end{align}
\end{subequations}
The local approximation of information leakage $\I \left( S; Z \right)$ derivation follows similar lines as \eqref{Eq:I_UZ_local_proof}.
\end{proof}

\vspace{1pt}


\bibliographystyle{IEEEtran}
\bibliography{references}

\end{document}